\documentclass[journal]{IEEEtran}

\usepackage{amsmath}
\usepackage{amssymb}
\usepackage{algorithm}
\usepackage{amsthm}
\usepackage{graphicx}
\usepackage{subfig}
\usepackage{enumerate}
\usepackage[noadjust]{cite}
\usepackage{commath}
\usepackage{float}
\usepackage{multicol}
\usepackage{color}
\usepackage{multirow}
\usepackage{array}
\usepackage{flushend}

\newtheorem{thm}{Theorem}[section]
\newtheorem{lem}{Lemma}[section]
\newtheorem{cor}{Corollary}[section]
\newtheorem{conj}{Conjecture}[section]

\newcommand{\aaa}{\mathbf{a}}
\newcommand{\bbb}{\mathbf{b}}
\newcommand{\cc}{\mathbf{c}}

\newcommand{\x}{\mathbf{x}}

\newcommand{\z}{\mathbf{z}}
\newcommand{\f}{\mathbf{f}}
\newcommand{\g}{\mathbf{g}}
\newcommand{\h}{\mathbf{h}}
\newcommand{\bl}{\mathbf{l}}
\newcommand{\s}{\mathbf{s}}

\newcommand{\vv}{\mathbf{v}}
\newcommand{\w}{\mathbf{w}}

\newcommand{\A}{\mathbf{a}}

\newcommand{\D}{\mathbf{D}}

\newcommand{\GG}{\mathbf{G}}
\newcommand{\HH}{\mathbf{H}}
\newcommand{\SSSS}{\mathbf{S}}

\newcommand{\Y}{\mathbf{Y}}
\newcommand{\Z}{\mathbf{Z}}
\newcommand{\LL}{\mathbf{L}}

\newcommand{\W}{\mathbf{W}}
\newcommand{\X}{\mathbf{X}}

\newcommand{\rank}{ \textrm{rank} }
\newcommand{\for}{\quad \textrm{for} \quad}

\newcommand{\trace}{ \textrm{trace} }

\newcommand{\edit}{}

\begin{document}


\title{STFT Phase Retrieval: Uniqueness Guarantees and Recovery Algorithms}
\author{
\begin{tabular}[t]{c@{\extracolsep{5em}}c@{\extracolsep{5em}}c} 
Kishore Jaganathan$^\dagger$ & Yonina C. Eldar$^\ddagger$  & Babak Hassibi$^\dagger$
\end{tabular}
\\
$^\dagger$Department of Electrical Engineering, Caltech \\
$^\ddagger$Department of Electrical Engineering, Technion, Israel Institute of Technology
\thanks{Copyright (c) 2014 IEEE. Personal use of this material is permitted. However, permission to use this material for any other purposes must be obtained from the IEEE by sending a request to pubs-permissions@ieee.org}  
\thanks{K. Jaganathan and B. Hassibi were supported in part by the National Science Foundation under grants CCF-0729203, CNS-0932428 and CIF-1018927, by the Office of Naval Research under the MURI grant N00014-08-1-0747, and by a grant from Qualcomm Inc.}
\thanks{Y. C. Eldar was supported in part by the European Union's Horizon 2020 Research and Innovation Program through the ERC-BNYQ Project, and in part by the Israel Science Foundation under Grant 335/14.}   
}

\maketitle

\begin{abstract}
The problem of recovering a signal from its Fourier magnitude is of paramount importance in various fields of engineering and applied physics. Due to the absence of Fourier phase information, some form of additional information is required in order to be able to uniquely, efficiently and robustly identify the underlying signal. Inspired by practical methods in optical imaging, we consider the problem of signal reconstruction from the Short-Time Fourier Transform (STFT) magnitude. We first develop conditions under which the STFT magnitude is an almost surely unique signal representation. We then consider a semidefinite relaxation-based algorithm (STliFT) and provide recovery guarantees. Numerical simulations complement our theoretical analysis and provide directions for future work.
\end{abstract}



\begin{figure*}
\begin{center}
\includegraphics[scale=0.45]{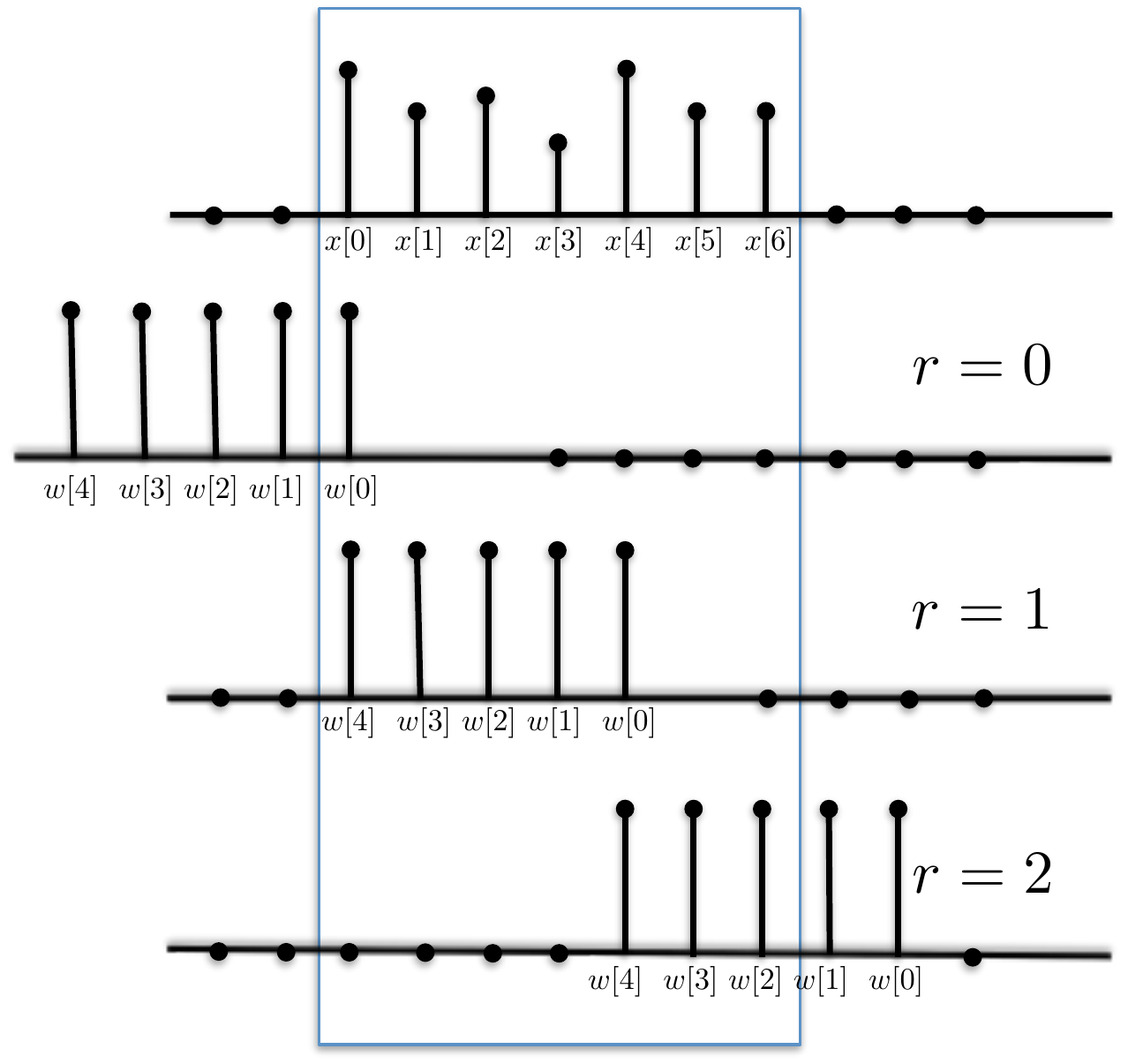} \hspace{1cm} \includegraphics[scale=0.45]{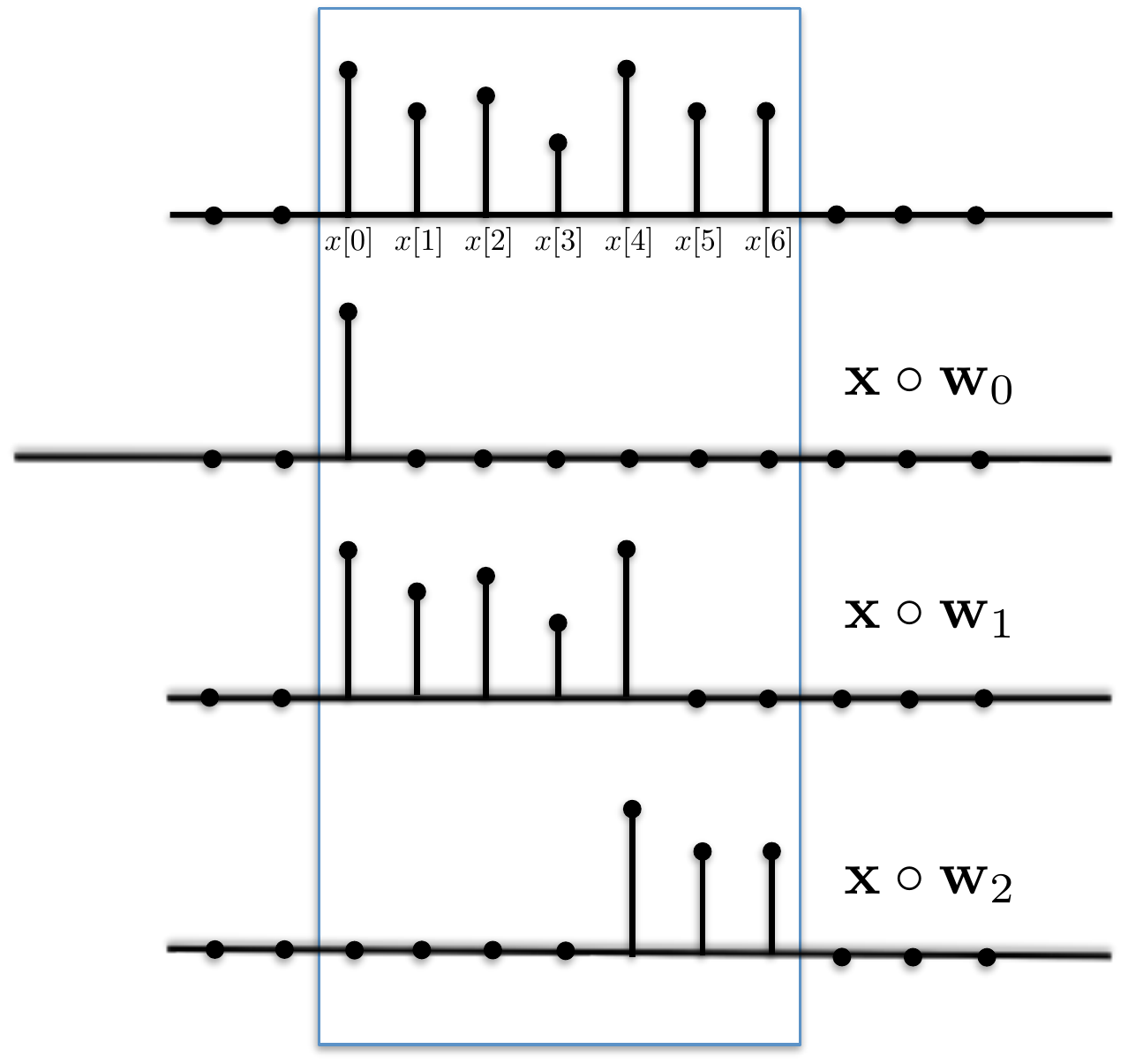}
\end{center}
\caption{Sliding window interpretation of the STFT for $N=7$, $W = 5$ and $L=4$. The shifted window overlaps with the signal for $3$ shifts, and hence $R=3$ short-time sections are considered.}
\label{fig:STFTexample}
\end{figure*}

\section{Introduction} 

In many physical measurement systems, the measurable quantity is the magnitude-square of the Fourier transform of the underlying signal. The problem of reconstructing a signal from its Fourier magnitude is known as phase retrieval \cite{patt1, patt2}. This reconstruction problem is one with a rich history and occurs in many areas of engineering and applied physics such as optics \cite{walther}, X-ray crystallography \cite{millane}, astronomical imaging \cite{dainty}, speech recognition \cite{rabiner}, computational biology \cite{stef}, blind channel estimation \cite{baykal} and more. We refer the readers to \cite{gerchberg, fienup, bauschke} for a comprehensive survey of classical approaches. Recent reviews can be found in \cite{eldarmagazine, kishorebc}.

It is well known that phase retrieval is an ill-posed problem \cite{hofstetter}. In order to be able to uniquely identify the underlying signal, various methods have been explored, which can be broadly classified into two categories: (i) {\em Additional prior information}: common approaches include bounds on the support of the signal \cite{gerchberg, fienup, bauschke} and sparsity constraints \cite{vetterli, mukherjee, candespr, eldar, eldar0, eldargespar, kishoreorig, kishorej}. (ii) {\em Additional magnitude-only measurements}: popular examples include the use of structured illuminations and masks \cite{liu, pfeiffer, popov, candespr, mahdi, gross, kishoreprm}, and Short-Time Fourier Transform (STFT) magnitude measurements \cite{trebino, yang, oppenheim, rodenburg, humphry, nawab, eldarstft, kishorestft, lim}.

We consider STFT phase retrieval, which is the problem of reconstructing a signal from its STFT magnitude. In some applications of phase retrieval, it is easy to obtain such measurements. One example is Frequency Resolved Optical Gating (FROG), which is a general method for measuring ultrashort laser pulses \cite{trebino}. Fourier ptychography \cite{humphry, rodenburg, yang}, a technology which has enabled X-ray, optical and electron microscopy with increased spatial resolution without the need for advanced lenses, is another popular example. In applications such as speech processing, it is natural to work with the STFT instead of the Fourier transform as the spectral content of speech changes over time \cite{oppenheim}. The key idea, when using STFT measurements, is to introduce redundancy in the magnitude-only measurements by maintaining a substantial overlap between adjacent short-time sections. This mitigates the uniqueness and algorithmic issues of phase retrieval. 

In this work, our contribution is two-fold:

{\em (i) Uniqueness guarantees}: Researchers have previously developed conditions under which the STFT magnitude uniquely identifies signals (up to a global phase). However, either prior information on the signal is assumed in order to provide the guarantees, or the guarantees are limited. For instance, the results provided in \cite{nawab} require exact knowledge of a small portion of the underlying signal. In \cite{eldarstft}, the guarantees developed are for the setup in which adjacent short-time sections differ in only one index. These limitations are primarily due to a small number of adversarial signals which cannot be uniquely identified from their STFT magnitude. Here, in contrast, we develop conditions under which the STFT magnitude is an {\em almost surely} unique signal representation. In particular, we show that, with the exception of a set of signals of measure zero, non-vanishing signals can be uniquely identified (up to a global phase) from their STFT magnitude if adjacent short-time sections overlap (Theorem \ref{STFTUN}). We then extend this result to incorporate sparse signals which have a limited number of consecutive zeros (Corollary \ref{STFTUS}). 

{\em (ii) Recovery algorithms:} Researchers have previously developed efficient iterative algorithms based on classic optimization frameworks to solve the STFT phase retrieval problem. {\edit Examples include the Griffin-Lim (GL) algorithm \cite{lim} and STFT-GESPAR for sparse signals \cite{eldarstft}. While these techniques work well in practice, they do not have theoretical guarantees}. In \cite{kishorestft} and \cite{sun}, a semidefinite relaxation-based STFT phase retrieval algorithm, called STliFT (see Algorithm \ref{algo:STliFT} below), was proposed. In this work, we conduct extensive numerical simulations and provide theoretical guarantees for STliFT. {\edit In particular, we conjecture that STliFT can recover most non-vanishing signals (up to a global phase) from their STFT magnitude if adjacent short-time sections differ in at most half the indices (Conjecture \ref{conj}). When this condition is satisfied, we argue that one can super-resolve (i.e., discard high frequency measurements) and reduce the number of measurements to $(4 + {o}(1) )N$, where $N$ is the length of the complex signal. Therefore, STliFT recovers most non-vanishing signals uniquely, efficiently and robustly, using an order-wise optimal number of phaseless measurements.

We prove this conjecture for the setup in which the exact knowledge of a small portion of the underlying signal is available (Theorem \ref{lhalfthm}). For particular choices of STFT parameters, this portion vanishes asymptotically, due to which this setup is asymptotically reasonable. We also prove this conjecture for the case in which adjacent short-time sections differ in only one index (Theorem \ref{l1thm}). We then extend these results to incorporate sparse signals which have a limited number of consecutive zeros (Corollary \ref{lhalfthms}).}

The rest of the paper is organized as follows. In Section 2, we mathematically formulate STFT phase retrieval and establish our notation. We present uniqueness guarantees in Section 3. Section 4 considers the STliFT algorithm and provides recovery guarantees. Numerical simulations are presented in Section 5. Section 6 concludes the paper.

\section{Problem Setup}

Let $\x = ( x[0] , x[1] , \ldots , x[N-1] )^T$ be a signal of length $N$ and $\w = ( w[0] , w[1] , \ldots , w[W-1] )^T$ be a window of length $W$. The STFT of  $\x$ with respect to $\w$, denoted by $\Y_w$, is defined as:
\begin{equation}
Y_w[ m , r ] = \sum_{n = 0}^{N-1} x[ n ] w[ rL - n ] e^{- i 2 \pi \frac{ m n } { N } }
\end{equation}
for $0 \leq m \leq N - 1$ and $0 \leq r \leq R-1$, where the parameter $L$ denotes the separation in time between adjacent short-time sections and the parameter $R=\big\lceil\frac{N+W-1}{L}\big\rceil$ denotes the number of short-time sections considered.


The STFT can be interpreted as follows: Suppose $\w_r$ denotes the signal obtained by shifting the flipped window $\w$ by $rL$ time units (i.e., $w_r[n] = w[rL - n]$) and $\circ$ is the Hadamard (element-wise) product operator. The $r$th column of  $\Y_w$, for $0 \leq r \leq R - 1$, corresponds to the $N$ point DFT of $\x \circ \w_r$. In essence, the window is flipped and slid across the signal (see Figure \ref{fig:STFTexample} for a pictorial representation), and $\Y_w$ corresponds to the Fourier transform of the windowed signal recorded at regular intervals. This interpretation is known as the {\em sliding window} interpretation. 

Let $\Z_w$ be the $N \times R$ measurements corresponding to the magnitude-square of the STFT of $\x$ with respect to $\w$ so that $Z_w[m,r] = \abs{Y_w[m,r]}^2$. Let $\W_r$, for $0 \leq r \leq R-1$, be the $N \times N$ diagonal matrix with diagonal elements $( w_r[0], w_r[1], \ldots , w_r[N-1]) $. STFT phase retrieval can be mathematically stated as: 
\begin{align}
\label{STFTPRf}
&\textrm{find} \hspace{1.9cm} \x  \\ 
\nonumber & \textrm{subject to} \hspace{1cm} Z_w[m,r] = \abs{\left<\f_m,\W_r\x\right>}^2
\end{align}
for $0 \leq m \leq N - 1$ and $0 \leq r \leq R-1$, where $\f_m$ is the conjugate of the $m$th column of the $N$ point DFT matrix and $\left<.,.\right>$ is the inner product operator. {\edit In fact, STFT phase retrieval can be equivalently stated by only considering the measurements corresponding to $0 \leq r \leq R-1$ and $1 \leq m \leq M$, for any parameter $M$ satisfying $2W \leq M \leq N$ (see Section \ref{appD} for details). This equivalence significantly reduces the number of measurements when $W \ll N$, which is typically the case in practical methods. In Section 4, we further reduce the number of measurements per short-time section through super-resolution. In particular, we consider the setup with $2L \leq W \leq \frac{N}{2}$ and $4L \leq M \leq N$.}


We use the following definitions: A signal $\x$ is said to be non-vanishing if $x[n] \neq 0$ for all $0 \leq n \leq N -1$. Similarly, a window $\w$ is said to be non-vanishing if $w[n] \neq 0$ for all $0 \leq n \leq W-1$. Further, a signal $\x$ is said to be sparse if it is not non-vanishing, i.e., $x[n] = 0$ for at least one $0 \leq n \leq N-1$.


\begin{table*}[t]
\small
\begin{center}
\begin{tabular}{|>{\centering\arraybackslash}m{20em}|p{30em}|}
\hline
\multirow{3}{*}{} & {Uniqueness if the first $L$ samples are known a priori, $2L \leq W \leq \frac{N}{2}$ and $\w$ is non-vanishing \cite{nawab}}\\
\cline{2-1}
\parbox{20em}{\centering Non-vanishing signals \\ \{$x[n] \neq 0$ for all $0 \leq n \leq N-1$\}}& {Uniqueness up to a global phase if $L=1$, $2 \leq W \leq \frac{N+1}{2}$, $W-1$ coprime with $N$ and mild conditions on $\w$ \cite{eldarstft}} \\
\cline{2-1}
& {Uniqueness up to a global phase for almost all signals if $L < W \leq \frac{N}{2}$ and $\w$ is non-vanishing [This work]} \\
\hline
\multirow{3}{*}{}& {Uniqueness for signals with at most $W-2L$ consecutive zeros if the first $L$ samples, starting from the first non-zero sample, are known a priori, $2L \leq W \leq \frac{N}{2}$ and $\w$ is non-vanishing \cite{nawab}}\\
\cline{2-1}
\parbox{20em}{\centering Sparse signals \\ \{$x[n] = 0$ for at least one $0 \leq n \leq N-1$\}}&{No uniqueness for most signals with $W$ consecutive zeros \cite{eldarstft}} \\
\cline{2-1}
& {Uniqueness up to a global phase and time-shift for almost all signals with less than $\min\{W-L, L\}$ consecutive zeros if $L < W \leq \frac{N}{2}$ and $\w$ is non-vanishing [This work]} \\
\hline
\end{tabular}
\caption{Uniqueness results for STFT phase retrieval (i.e., $2W \leq M \leq N$).}
\end{center}
\end{table*}

\section{Uniqueness Guarantees} 

In this section, we review existing results regarding the uniqueness of STFT phase retrieval and present our uniqueness guarantees. The results are summarized in Table I.

In STFT phase retrieval, the global phase of the signal cannot be determined due to the fact that signals $\x$ and $e^{i \phi}\x$, for any $\phi$, always have the same STFT magnitude regardless of the choice of $\{\w,L\}$. In contrast, in classic phase retrieval, signals which differ from each other by a global phase, time-shift and/ or conjugate-flip (together called {\em trivial ambiguities}) cannot be distinguished from each other as they have the same Fourier magnitude \cite{vetterli, kishoreorig, eldarmagazine}. 


Observe that $L<W$ is a necessary condition in order to be able to uniquely identify most signals. If $L > W$, then the STFT magnitude does not contain any information from some locations of the signal. When $L=W$,  adjacent short-time sections do not overlap and hence STFT phase retrieval is equivalent to a series of non-overlapping instances of classic phase retrieval. Since there is no way of determining the relative phase, time-shift or conjugate-flip between the windowed signals corresponding to the various short-time sections, most signals cannot be uniquely identified. For example, suppose $\{\w,L\}$ is chosen such that $L=W=2$ and $w[n] = 1$ for all $0 \leq n \leq W-1$. Consider the signal $\x_1 = ( 1 , 2 , 3 )^T$ of length $N=3$. Signals $\x_1$ and $\x_2 = ( 1 , -2 , -3 )^T$ have the same STFT magnitude. In fact, more generally, signals $\x_1$ and $( 1 , e^{{i} \phi}2 , e^{{i} \phi}3 )^T$, for any $\phi$, have the same STFT magnitude.

\subsection{Non-vanishing signals}

For some specific choices of $\{\w,L\}$, it has been shown that all non-vanishing signals can be uniquely identified from their STFT magnitude up to a global phase. In \cite{eldarstft}, it is proven that the STFT magnitude uniquely identifies non-vanishing signals up to a global phase for $L=1$ if the window $\w$ is chosen such that the $N$ point DFT of $(|w[0]|^2, |w[1]|^2, \ldots, |w[N-1]|^2)$ is non-vanishing, $2 \leq W  \leq \frac{N+1}{2}$ and $W-1$ is coprime with $N$. In \cite{nawab}, the authors prove that if the first $L$ samples are known a priori, then the STFT magnitude can uniquely identify non-vanishing signals for any $L$ if the window $\w$ is chosen such that it is non-vanishing and $2L \leq W \leq \frac{N}{2}$. 

In this work, we prove the following result for non-vanishing signals:

\begin{thm}
\label{STFTUN}
Almost all non-vanishing signals can be uniquely identified (up to a global phase) from their STFT magnitude if $\{\w, L, M\}$ satisfy 
\begin{enumerate}[(i)]
\item $\w$ is non-vanishing
\item $L < W \leq \frac{N}{2}$
\item $2W \leq M \leq N$.
\end{enumerate}
\end{thm}
\begin{proof}
The proof is based on a technique commonly known as {\em dimension counting}. The outline is as follows {\edit (see Section \ref{appA} for details)}: 

Consider the short-time sections $r$ and $r+1$. Since adjacent short-time sections overlap (due to $L < W$), there exists at least one index, say $n_0$, where both $\x \circ \w_{r}$ and $\x \circ \w_{r+1}$ have non-zero values. 

Since $W \leq \frac{N}{2}$, there can be at most $2^W$ distinct windowed signals $\x \circ \w_r$ (up to a phase) that have the same Fourier magnitude \cite{hofstetter}. Consequently, $\abs{x[n_0]}$ is restricted to $2^W$ values by the $r$th column of the STFT magnitude (let $\mathcal{S}_r$ denote the set of these values). Similarly, $\abs{x[n_0]}$ is restricted to $2^W$ values by the $r+1$th column of the STFT magnitude (denote the set of these values by $\mathcal{S}_{r+1}$).

By construction, $\mathcal{S}_r \cap \mathcal{S}_{r+1} \neq \phi$ as the STFT magnitude is generated by an underlying signal $\x_0$, i.e., $\abs{x_0[n_0]} \in \mathcal{S}_r \cap \mathcal{S}_{r+1}$. Using Lemma \ref{stftproplem} and Theorem \ref{dimlem}, we show that, for almost all non-vanishing signals, $\mathcal{S}_r \cap \mathcal{S}_{r+1}$ has cardinality one. In other words, $\x \circ \w_r$ is uniquely identified (up to a phase) almost surely. 

Since adjacent short-time sections overlap, non-vanishing signals are uniquely identified up to a {\em global} phase from the knowledge of $\x \circ \w_r$ (up to a phase) for $0 \leq r \leq R-1$ if $\w$ is non-vanishing.
\end{proof}

\subsection{Sparse signals}

While the aforementioned results provide guarantees for non-vanishing signals, they do not say anything about sparse signals. {\edit Reconstruction of sparse signals involves certain challenges which are not encountered in the reconstruction of non-vanishing signals.}

The following example is provided in \cite{eldarstft} to show that the time-shift ambiguity cannot be resolved for some classes of sparse signals and some choices of $\{\w, L\}$: Suppose $\{\w, L\}$ is chosen such that $L \geq 2$, $W$ is a multiple of $L$ and $w[n] = 1$ for all $0 \leq n \leq W-1$. Consider a signal $\x_1$ of length $N \geq L+1$ such that it has non-zero values only within an interval of the form $[(t -1)L + 1, (t-1)L+L-p]  \subset [0,N-1]$ for some integers $1 \leq p \leq L-1$ and $t \geq 1$. The signal $\x_2$ obtained by time-shifting $\x_1$ by $q \leq p$ units (i.e., $x_2[i] = x_1[i - q]$) has the same STFT magnitude. The issue with this class of sparse signals is that the STFT magnitude is identical to the Fourier magnitude because of which the time-shift and conjugate-flip ambiguities cannot be resolved. 

It is also shown that some sparse signals cannot be uniquely recovered even up to the trivial ambiguities for some choices of  $\{\w, L\}$ using the following example: Consider two non-overlapping intervals $[u_1, v_1] , [u_2,v_2] \subset [0, N-1]$ such that $u_2 - v_1 > W$, and take a signal $\x_1$ supported on $[u_1,v_1]$ and $\x_2$ supported on $[u_2,v_2]$. The magnitude-square of the STFT of $\x_1+\x_2$ and of $\x_1-\x_2$ are equal for any choice of $L$. The difficulty with this class of sparse signals is that the two intervals with non-zero values are separated by a distance greater than $W$ because of which there is no way of establishing relative phase using a window of length $W$. 

These examples demonstrate the fact that sparse signals are harder to recover than non-vanishing signals in this setup. Since the aforementioned issues are primarily due to a large number of consecutive zeros, the uniqueness guarantees for non-vanishing signals have been extended to incorporate sparse signals with limits on the number of consecutive zeros. In \cite{nawab}, it was shown that if $L$ consecutive samples, starting from the first non-zero sample, are known a priori, then the STFT magnitude can uniquely identify signals with less than $W-2L$ consecutive zeros for any $L$ if the window $\w$ is chosen such that it is non-vanishing and $2L \leq W \leq \frac{N}{2}$. 

Below, we extend Theorem \ref{STFTUN} to prove the following result for sparse signals: 

\begin{cor}
\label{STFTUS}
Almost all sparse signals with less than $\min\{W-L, L\}$ consecutive zeros can be uniquely identified (up to a global phase and time-shift) from their STFT magnitude if $\{\w, L, M\}$ satisfy 
\begin{enumerate}[(i)]
\item $\w$ is non-vanishing
\item $L < W \leq \frac{N}{2}$
\item $2W \leq M \leq N$.
\end{enumerate}
\end{cor}
\begin{proof}
The $\min\{W-L,L\}$ bound on consecutive zeros ensures the following: For sufficient pairs of adjacent short-time sections, there is at least one index among the overlapping and non-overlapping indices respectively, where the underlying signal has a non-zero value. We refer the readers to Section \ref{appA2} for details.
\end{proof}

\section{Recovery Algorithms}

The classic alternating projection algorithm to solve phase retrieval \cite{gerchberg} has been adapted to solve STFT phase retrieval by Griffin and Lim \cite{lim}.  To this end, STFT phase retrieval is reformulated as the following least-squares problem:
\begin{equation}
\label{STFTleastsquares}
\min_\x \quad \sum_{r=0}^{R-1}  \sum_{m=0}^{N-1} { \left( Z_w[ m , r ] - \abs{\left<\f_m,\W_r\x\right>}^2 \right)^2 }.
\end{equation}
The Griffin-Lim (GL) algorithm attempts to minimize this objective by starting with a random initialization and imposing the time domain and STFT magnitude constraints alternately using projections. The objective is shown to be monotonically decreasing as the iterations progress. An important feature of the GL algorithm is its empirical ability to converge to the global minimum when there is substantial overlap between adjacent short-time sections. However, no theoretical recovery guarantees are available. To establish such guarantees, we rely on a semidefinite relaxation approach.

\subsection{Semidefinite relaxation-based algorithm}

Semidefinite relaxation has enjoyed considerable success in provably and stably solving several quadratic-constrained problems \cite{goemans, daspremont, candespl}. The steps to formulate such problems as a semidefinite program (SDP) are as follows: (i) Embed the problem in a higher dimensional space using the transformation $\X = \x\x^\star$, a process which converts the problem of recovering a signal with quadratic constraints into a problem of recovering a rank-one matrix with affine constraints. (ii) Relax the rank-one constraint to obtain a convex program. 

{\edit If the convex program has a unique solution $\X_0 = \x_0\x_0^\star$, then $\x_0$ is the unique solution to the quadratic-constrained problem (up to a global phase).  Many recent results in related problems like {\em generalized phase retrieval} \cite{candespl} and {\em phase retrieval using random masks} \cite{mahdi, gross} suggest that one can provide conditions, which when satisfied, ensure that the convex program has a unique solution $\X_0 = \x_0\x_0^\star$.} 

A semidefinite relaxation-based STFT phase retrieval algorithm, called STliFT, was explored in \cite{kishorestft} and \cite{sun}. The details of the algorithm are provided in Algorithm \ref{algo:STliFT}. In the following, we develop conditions on $\{\x_0, \w, L\}$ which ensure that the convex program (\ref{STFTPRR}) has $\X_0 = \x_0\x_0^\star$ as the unique solution. Consequently, under these conditions, STliFT uniquely recovers the underlying signal up to a global phase. 


\begin{algorithm}
\caption{STliFT}
\label{algo:STliFT}
\textbf{Input:} STFT magnitude measurements $Z_w[m,r]$ for $1 \leq m \leq M$ and $0 \leq r \leq R-1$, $\{\w, L\}$. \\
\textbf{Output:} Estimate $\hat{\x}$ of the underlying signal $\x_0$.

\begin{itemize}
\item Obtain  $\hat{\X}$ by solving:
\begin{align}
\label{STFTPRR}
&\textrm{minimize} \hspace{1cm} \trace( \X )  \\
\nonumber & \textrm{subject to} \hspace{0.9cm} {Z_w[m,r]} =  \trace(\W_r^\star\f_m \f_m^\star\W_r\X) \\
& \hspace{2.3cm} \X \succcurlyeq 0 \nonumber 
\end{align}
for $1 \leq m \leq M$ and $0 \leq r \leq R-1$.
\item Return $\hat{\x}$, where $\hat{\x} \hat{\x}^\star$ is the best rank-one approximation of $\hat{\X}$.
\end{itemize}
\end{algorithm}

Based on extensive numerical simulations, we conjecture the following:

\begin{conj}
The convex program (\ref{STFTPRR}) has a unique solution $\X_0=\x_0\x_0^\star$, for most non-vanishing signals $\x_0$, if 
\begin{enumerate}[(i)]
\item $\w$ is non-vanishing
\item $2L \leq W \leq \frac{N}{2}$
\item $4L \leq M \leq N$.
\end{enumerate}
\label{conj}
\end{conj}

The number of phaseless measurements considered can be calculated as follows: The total number of short-time sections is $\lceil{\frac{N+W-1}{L}}\rceil$. For each short-time section, $M = 4L$ phaseless measurements are sufficient. Hence, the total number of phaseless measurements is $\lceil{\frac{N+W-1}{L}}\rceil \times 4L \leq 4 \left({N+W} \right) + 2W$. Consequently, when $W = o(N)$, this number is $( 4 + o(1))N$, which is order-wise optimal. In fact, in generalized phase retrieval, it is conjectured that $( 4 - o(1))N$ phaseless measurements are necessary \cite{balan}.

{\edit The proof techniques used in \cite{mahdi} and \cite{candespl} are not applicable in the STFT setup. In \cite{mahdi} and  \cite{candespl}, the measurement vectors are chosen from a random distribution such that they satisfy the restricted isometry property. Furthermore, the randomness in the measurement vectors is used to construct approximate dual certificates based on concentration inequalities. In the STFT setup, testing whether the given measurement vectors satisfy the restricted isometry property is difficult. Also, due to the lack of randomness in the measurement vectors, a different approach is required to construct dual certificates. 

In the following, we develop a proof technique for the STFT setup, and use it to prove Conjecture \ref{conj}, with additional assumptions.}

\begin{figure*}
\begin{center}
\subfloat[{$N=32$, $M=4L$, and various choices of $\{L,W\}$.}]{\includegraphics[scale = 0.41]{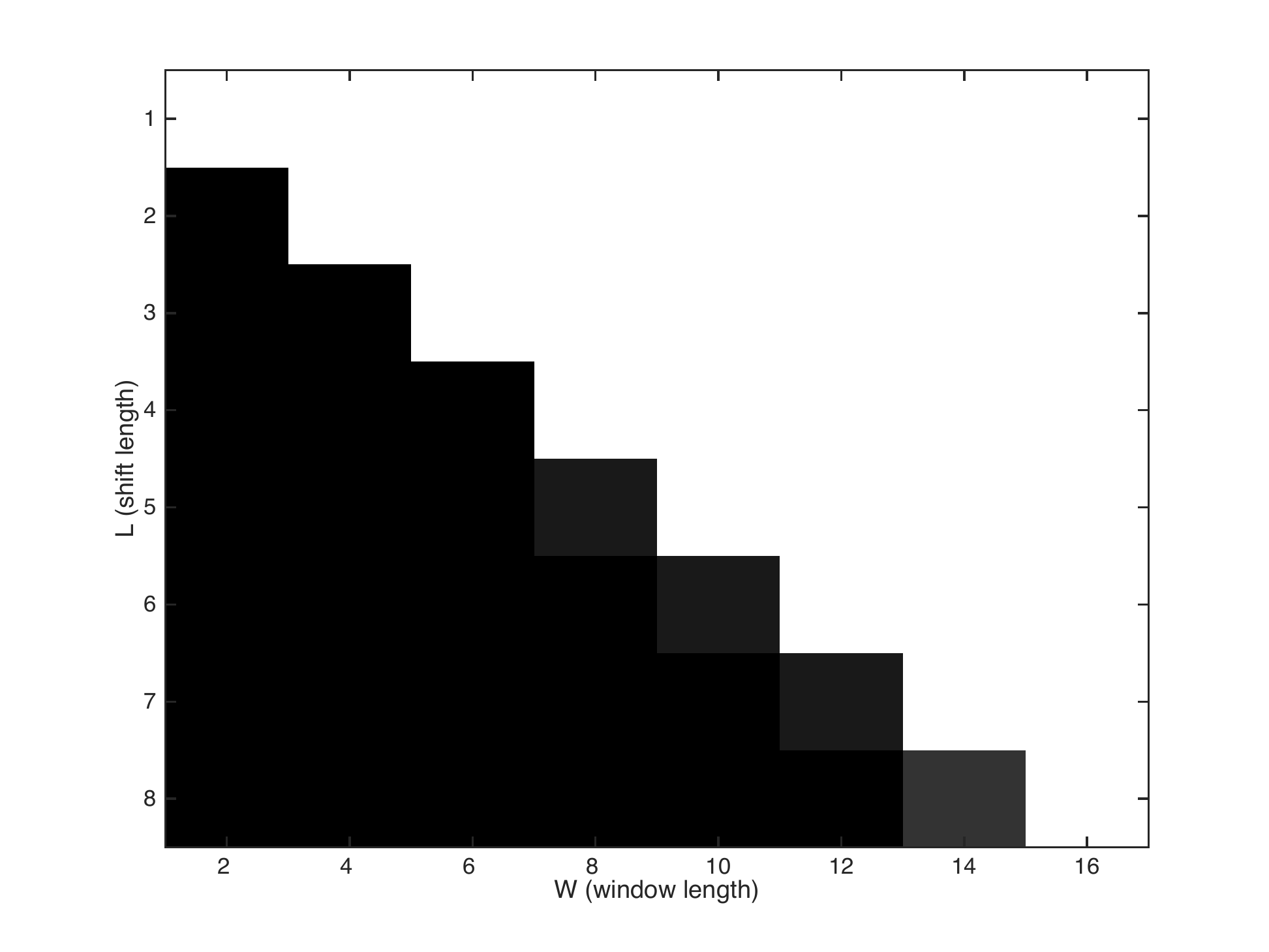} \label{fig:successprobabilitystft}
}  \hspace{1cm}
\subfloat[ {$N=32$, $W=16$, and various choices of $\{L, M\}$ (demonstrates super-resolution).}]{\includegraphics[scale = 0.41]{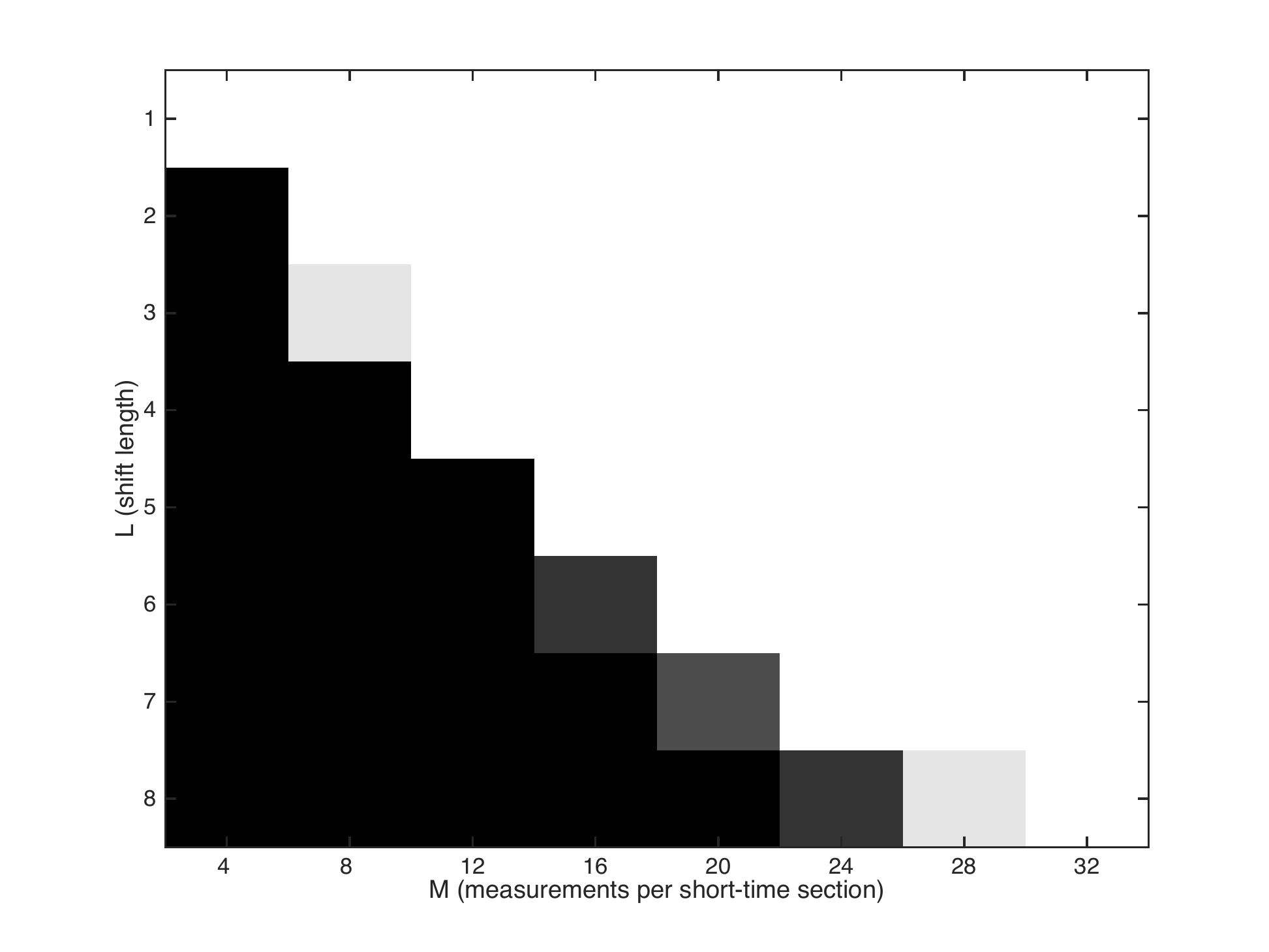} \label{fig:successprobabilitystftsr}}  
\end{center}
\caption{Probability of successful recovery using STliFT in the noiseless setting (white region: success with probability $1$, black region: success with probability $0$).}
\end{figure*}

\begin{thm}
The convex program (\ref{STFTPRR}) has a unique feasible matrix $\X_0=\x_0\x_0^\star$, for almost all non-vanishing signals $\x_0$, if 
\begin{enumerate}[(i)]
\item $\w$ is non-vanishing
\item $2L \leq W \leq \frac{N}{2}$
\item $4L \leq M \leq N$
\item $x_0[n]$ for $0 \leq n \leq \big\lfloor{\frac{L}{2}\big\rfloor}$ is known a priori. 
\end{enumerate}
\label{lhalfthm}
\end{thm}
\begin{proof}
See Section \ref{appC}.
\end{proof}

{\edit While it is sufficient to show that (\ref{STFTPRR}) has a unique solution $\X_0=\x_0\x_0^\star$, observe that Theorem \ref{lhalfthm} ensures that (\ref{STFTPRR}) has a unique feasible matrix. This is a stronger condition, and as a consequence, the choice of the objective function does not matter in the noiseless setting. While this might suggest that the requirements of the setup are strong, we argue that it is not the case. In fact, this phenomenon is also observed in generalized phase retrieval (Section $1.3$ in \cite{candespl}) and phase retrieval using random masks (Theorem $1.1$ in \cite{mahdi}).}

{\edit Theorem \ref{lhalfthm} assumes prior knowledge of the first $\lceil{\frac{L}{2}}\rceil$ samples, i.e., half of the second short-time section is required to be known a priori. This is not a lot of prior information if $W \ll N$, which is typically the case. When $W=o(N)$, the fraction of the signal that is required to be known a priori is less than $\frac{W}{N}$, which tends to $0$ as $N \rightarrow \infty$.}

\begin{thm}
The convex program (\ref{STFTPRR}) has a unique feasible matrix $\X_0=\x_0\x_0^\star$, for almost all non-vanishing signals $\x_0$, if
\begin{enumerate}[(i)]
\item $\w$ is non-vanishing
\item $2 \leq W \leq \frac{N}{2} $
\item $4 \leq M \leq N$
\item $L = 1$.
\end{enumerate}
\label{l1thm}
\end{thm}
\begin{proof}
This is a direct consequence of Theorem \ref{lhalfthm}. The value of $\abs{x_0[0]}$ (and hence $x_0[0]$, without loss of generality) can be inferred from the STFT magnitude if $L=1$.
\end{proof}

When $L=1$, the number of phaseless measurements is $4(N + W)$, which is again order-wise optimal. For example, when $W=2$, at most $4N + 8$ phaseless measurements are considered. Unlike Theorem \ref{lhalfthm}, no prior information is necessary.


Theorems \ref{lhalfthm} and \ref{l1thm} can be seamlessly extended to incorporate sparse signals:
\begin{cor}
The convex program (\ref{STFTPRR}) has a unique feasible matrix $\X_0=\x_0\x_0^\star$, for almost all sparse signals $\x_0$ which have at most $W-2L$ consecutive zeros, if
\begin{enumerate}[(i)]
\item $\w$ is non-vanishing
\item $2L \leq W \leq \frac{N}{2} $
\item $4L \leq M \leq N$
\item Either $L=1$ or $x_0[n]$ for $i_0 \leq n < i_0 + L$ is known a priori, where $i_0$ is the smallest index such that $x_0[i_0] \neq 0$. 
\end{enumerate}
\label{lhalfthms}
\end{cor}
\begin{proof}
See Section \ref{appC}.
\end{proof}



\subsection{Noisy Setting}

In practice, the measurements are contaminated by additive noise, i.e., the measurements are of the form
\begin{equation}
{Z_w[m,r]} =  \abs{\left<\f_m,\W_r\x\right>}^2+ z[m,r] \nonumber
\end{equation} 
for $1 \leq m \leq M$ and $0 \leq r \leq R-1$, where $\z_r = (z[0,r], z[1,r] , \ldots , z[M-1,r])^T$ is the additive noise corresponding to the $r${th} short-time section and $4L \leq M \leq N$. STliFT, in the noisy setting, can be implemented as follows: Suppose $\| \z_r \|_2 \leq \eta$ for all $0 \leq r \leq R-1$. The constraints in the convex program (\ref{STFTPRR}) can be replaced by
\begin{equation}
{\sum_{m=1}^{M}\Big( Z_w[ m , r ] -  \trace(\W_r^\star\f_m \f_m^\star\W_r\X) \Big)^2} \leq \eta^2
\end{equation}
for $0 \leq r \leq R-1$. We recommend the use of trace minimization as the objective function. Numerical simulations strongly suggest that STliFT can recover most non-vanishing signals stably in the noisy setting under certain conditions. The details of the simulations are provided in the following section.

\section{Numerical Simulations}

In this section, we demonstrate the empirical abilities of STliFT using numerical simulations. 

In the first set of simulations, we evaluate the performance of STliFT as a function of window and shift lengths. We choose $N=32$, and vary $\{L,W\}$. For each choice of $\{L,W\}$, we consider $M = 4L$ phaseless measurements and perform $100$ trials. In every trial, we choose a random signal such that the values in each location are drawn from an i.i.d. standard complex normal distribution. We select the window $\w$ such that $w[n]=1$ for all $0 \leq n \leq W-1$. The probability of successful recovery as a function of $\{L,W\}$ is plotted in Fig. \ref{fig:successprobabilitystft}. 

Observe that STliFT successfully recovers the underlying signal with very high probability when $2L \leq W \leq \frac{N}{2}$ and fails with very high probability when $2L > W$. The choice of $\{L, W\}= \{ \frac{N}{4}, \frac{N}{2}\}$ uses only {\em six} short-time sections and STliFT recovers the underlying signal with very high probability, which, given the limited success of semidefinite relaxation-based algorithms in the Fourier phase retrieval setup, is very encouraging.

In the second set of simulations, we evaluate the performance of STliFT as a function of shift length and measurements per short-time section. We choose $N=32$ and $W = 16$, and vary $\{L,M\}$. For each choice of $\{L,M\}$, we perform $100$ trials as before. The probability of successful recovery as a function of $\{L,M\}$ is plotted in Fig. \ref{fig:successprobabilitystftsr}.  Observe that recovery is successful even in the $4L \leq M <2W$ regime.

\begin{figure}
\begin{center}
\includegraphics[scale = 0.41]{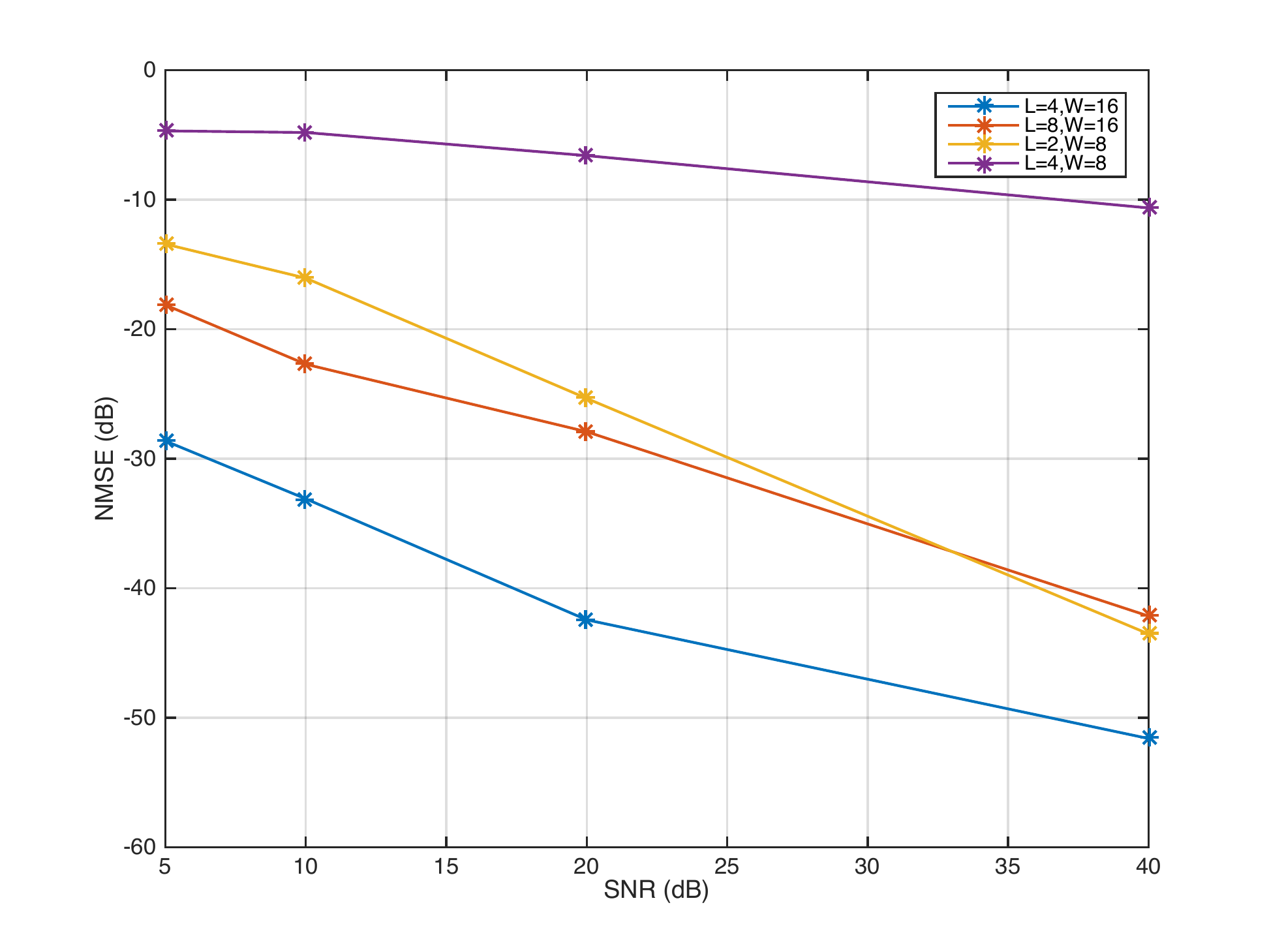} 
\end{center}
\caption{NMSE (dB) vs SNR (dB) using STliFT in the noisy setting for $N=32$, $M=2W$.}
\label{fig:errorstft}
\end{figure}

In the third set of simulations, we evaluate the performance of STliFT in the noisy setting. We choose $M=2W$, the rest of the parameters are the same as the first set of simulations. The normalized mean-squared error, given by
\begin{equation}
NMSE = \min\limits_{\abs{c} = 1 } \frac{ ||\x_0 - c \hat{\x}||^2 } { || \x_0 ||^2 },
\end{equation}
is plotted as a function of SNR in Fig. \ref{fig:errorstft}. The linear relationship between them shows that STliFT stably recovers the underlying signal in the presence of noise. Also, it can be observed that the choices of $\{W,L\}$ which correspond to significant overlap between adjacent short-time sections tend to recover signals more stably compared to values of $\{W,L\}$ which correspond to less overlap, which is not surprising.


\section{Conclusions and Future Directions}

In this work, we considered the STFT phase retrieval problem. We showed that, if $L < W \leq \frac{N}{2}$, then almost all non-vanishing signals can be uniquely identified from their STFT magnitude (up to a global phase), and extended this result to incorporate sparse signals which have less than $\min\{W-L,L\}$ consecutive zeros. 

For $2L \leq W \leq \frac{N}{2}$, we conjectured that most non-vanishing signals can be recovered (up to a global phase) by a semidefinite relaxation-based algorithm (STliFT). When $W=o(N)$, through super-resolution, we reduced the number of phaseless measurements to $(4 + o(1))N$. We proved this conjecture for the setup in which the first $\big\lfloor \frac{L}{2} + 1\big\rfloor$ samples are known, and for the case in which $L=1$. We argued that the additional assumptions are asymptotically reasonable when $W\ll N$, which is typically the case in practical methods. We then extended these results to incorporate sparse signals which have at most $W-2L$ consecutive zeros.

Natural directions for future study include a proof of this conjecture without the additional assumptions, and a stability analysis in the noisy setting. Also, a thorough analysis of the phase transition at $2L=W$ will provide a more complete characterization of STliFT.

{\bf Acknowledgements}: We would like to thank Mordechai Segev and Oren Cohen for introducing us to the STFT phase retrieval problem, and for many insightful discussions.



\section*{Appendix}

\section{Equivalent definition of stft phase retrieval}

\label{appD}

{\edit Since we consider $N$ point DFT and $W$ satisfies $W \leq \frac{N}{2}$, STFT phase retrieval can be equivalently stated in terms of the short-time autocorrelation $\A_w$ \cite{hofstetter}: 
\begin{align}
\label{STFTPRa}
&\textrm{find} \hspace{1.6cm} \x  \\
\nonumber & \textrm{subject to} \\
\nonumber & a_w[ m , r ] = { \sum_{n = 0}^{N-1-m} x[n] w[ rL - n ] x^\star[n+m]w^\star[rL-(n+m)]}
\end{align}
for $0 \leq m \leq N - 1$ and $0 \leq r \leq R-1$.

The knowledge of the short-time autocorrelation is sufficient for all the guarantees provided in this paper. Note that the $r${th} column of $\Z_w$ and the $r${th} column of $\aaa_w$ are Fourier pairs. Hence, for a particular $r$, if $Z_w[m,r]$ for $0 \leq m \leq N-1$ is available, then $a_w[m,r]$ for $0 \leq m \leq N-1$ can be calculated  by taking an inverse Fourier transform. The following lemma shows that $2W$ phaseless measurements per short-time section are sufficient to infer the short-time autocorrelation.

\begin{lem}
$Z_w[m,r]$ for $1 \leq m \leq 2W-1$ is sufficient to calculate $a_w[m,r]$ for $0 \leq m \leq N-1$. 
\end{lem}
\begin{proof}
If the window length is $W$, then $\aaa_w$ has non-zero values only in the interval $0 \leq m \leq W-1$ and $N - W + 1 \leq m \leq N-1$. Let $\bbb_w$ be the signal obtained by circularly shifting $\aaa_w$ by $W-1$ rows, so that $\bbb_w$ has non-zero values only in the interval $0 \leq m \leq 2W-2$. Since the submatrix of the $N$ point DFT matrix obtained by considering the first $2W-1$ columns and any $2W-1$ rows is invertible (the Vandermonde structure is retained), $Z_w[m,r]$ for $1 \leq m \leq 2W-1$ and $b_w[m,r]$ for  $0 \leq m \leq 2W-2$ are related by an invertible matrix. Note that $a_w[m,r]$ for $0 \leq m \leq N-1$ can be trivially calculated from $b_w[m,r]$ for  $0 \leq m \leq 2W-2$.
\end{proof}

Consequently, if the $N$ point DFT is used and $2W \leq M \leq N$ is satisfied, the affine constraints in (\ref{STFTPRR}) can be rewritten in terms of $\aaa_w$ and $\X$ as:
\begin{align}
\nonumber & a_w[ m , r ] = { \sum_{n = 0}^{N-1-m} X[n,n+m] w[ rL - n ]w^\star[rL-(n+m)] }. \nonumber
\end{align}}

\section{Proof of Theorem \ref{STFTUN}}

\label{appA}

The symbol $\equiv$ is used to denote equality up to a global phase and time-shift\footnote[4]{For non-vanishing signals, there is no ambiguity due to time-shift.}. We say that two signals $\x_1$ and $\x_2$ are distinct if $\x_1 \not\equiv \x_2$, and equivalent if $\x_1\equiv \x_2$. 

{\edit 
Let $\mathcal{P}$ denote the set of all distinct non-vanishing complex signals of length $N$. $\mathcal{P}$ is a manifold of dimension $2N-1$, i.e., $\mathcal{P}$ locally resembles a real $2N-1$ dimensional space. This can be seen as follows: In order to discard the global phase of non-vanishing signals, we can assume that $x[n_0]$ is real and positive at  one index $n_0$, without loss of generality. Hence, $x[n_0]$ can take any value in $\mathbb{R}_+$, and $x[n]$, for each $0 \leq n \leq N-1$ not equal to $n_0$, can take any value in $\mathbb{R}^2 \backslash \{0,0\}$, due to the one-to-one correspondence between $\mathbb{C}$ and $\mathbb{R}^2$.

Let $\mathcal{P}_c \subset \mathcal{P}$ be the set of distinct non-vanishing complex signals which cannot be uniquely identified from their STFT magnitude if $\w$ is chosen such that it is non-vanishing and $W \leq \frac{N}{2}$. We show that $\mathcal{P}_c$ has measure zero in $\mathcal{P}$. In order to do so, our strategy is as follows: 

We first characterize $\mathcal{P}_c$ using Lemma \ref{stftproplem}. In particular, we show that $\mathcal{P}_c$ is a finite union of images of continuously differentiable maps from $\mathbb{R}^{2N-2}$ to $\mathcal{P}$. Since $\mathcal{P}$ is a manifold of dimension $2N-1$, the following result completes the proof:
\begin{thm}[\cite{ortega}, Chapter 5]
\label{dimlem}
If $f: \mathbb{R}^{N_0} \rightarrow \mathbb{R}^{N_1}$ is a continuously differentiable map, then the image of $f$ has measure zero in $\mathbb{R}^{N_1}$, provided $N_0<N_1$.
\end{thm}

}

We use the following notation in this section: If $\g$ is a signal of length $l_g$, then $\g = ( g[0], g[1], \ldots , g[l_g-1] )^T$ such that $\{g[0], g[l_g-1]\} \neq 0$ and $g[n]=0$ outside the interval $[0, l_g-1]$. The vector $\tilde{\g}$ denotes the conjugate-flipped version of $\g$, i.e., $\tilde{\g} = ( g^\star[l_g-1], g^\star[l_g-2], \ldots , g^\star[0] )^T$. Let $u_r$ and $v_r$ denote the smallest and largest index  where the windowed signal $\x \circ \w_r$ has a non-zero value respectively.

\begin{lem}
Consider two signals $\x_1 \not \equiv \x_2$ of length $N$ which have the same STFT magnitude.  If the window $\w$ is chosen such that it is non-vanishing and $W \leq \frac{N}{2}$, then, for each $r$, there exists signals $\g_r$ and $\h_r$, of lengths $l_{gr}$ and $l_{hr}$ respectively, such that 
\begin{enumerate}[(i)]
\item $ \x_1 \circ {\w}_r \equiv \g_r \star \h_r, ~~\x_2 \circ {\w}_{r} \equiv \g_r \star \tilde{\h}_r$
\item $l_{gr} + l_{hr} - 1 = v_r - u_r + 1$
\item $g_r[l_{gr}-1] = 1$ and $\{ g_r[0] , h_r[0], h_r[l_{hr}-1] \}\neq 0$
\end{enumerate}
\label{stftproplem}
where $\star$ is the convolution operator. Further, there exists at least one $r$ such that
\begin{enumerate}
\item[(iv)] $l_{hr} \geq 2$, $h_r[0]$ is real and positive. 
\end{enumerate}
\end{lem}
\begin{proof}
In Lemma $7.1$ of \cite{kishorej}, it is shown that if two non-equivalent signals of length $N$ have the same Fourier magnitude and if the DFT dimension is at least $2N$ (this would imply that they have the same autocorrelation), then there exists signals $\g$ and $\h$, of lengths $l_g$ and $l_h$  respectively, such that one signal can be decomposed as $\g \star \h$ and the other signal can be decomposed as $\g \star \tilde{\h}$. 
{\edit For each $r$, the $r$th column of the STFT magnitude is equivalent to the Fourier magnitude of the windowed signal $\x \circ \w_r$. The DFT dimension is $N$, and the windowed signal length is  $v_r - u_r + 1$ (which is less than or equal to $\frac{N}{2}$). Since, for every $r$, $\x_1 \circ \w_r$ and $\x_2 \circ \w_r$ have the same Fourier magnitude, the aforementioned result proves (i).}

The conditions (ii) and (iii) are properties of convolution (see Lemma $7.1$ of \cite{kishorej} for details), and therefore hold for every $r$. 

Furthermore, if $l_{hr} = 1$ for all $0 \leq r \leq R-1$, then $\x_1 \equiv \x_2$. Hence, $l_{hr} \geq 2$ for at least one $r$. For this $r$, since $e^{i\phi_1}\x_1$ and $e^{i\phi_2}\x_2$ have the same STFT magnitude, $h_r[0]$ can be assumed to be real and positive without loss of generality. Hence, (iv) holds for at least one $r$.
\end{proof}


{\edit 

Consequently, for each $\x \in \mathcal{P}_c$, condition (iv) of Lemma \ref{stftproplem} holds for at least one $r$. Let $\mathcal{P}_{c}^{rl_rl_{r+1}} \subset \mathcal{P}_c$ denote the set of signals for which $l_{hr} = l_r \geq 2$ and $l_{h,r+1} = l_{r+1}$\footnote[2]{When $r=R$, we consider the short-time section $r-1$ instead of $r+1$. We show the detailed calculations for the case when short-time section $r + 1$ is considered, the arguments are symmetric for $r-1$.}. It suffices to show that for each $r$, $l_r$ and $l_{r+1}$, there exists a set $\mathcal{Q}_{c}^{rl_rl_{r+1}} \supseteq \mathcal{P}_{c}^{rl_rl_{r+1}}$, which is the image of a continuously differentiable map from $\mathbb{R}^{2N-2}$ to $\mathcal{P}$. 

We first show the arguments for the $L=W-1$ case as the expressions are simple and provide intuition for the technique. Then, we show the arguments for the $L < W-1$ case. 

${(i) ~ L = W-1:}$



The set $\mathcal{Q}_{c}^{rl_rl_{r+1}}$ is constructed as follows: Consider the variables $\{\g_r, \h_r, \g_{r+1}, \h_{r+1}\}$ satisfying $l_{hr}=l_r \geq 2$ and $l_{h,r+1} = l_{r+1}$, and $x[n]$ for $n \in [0, u_r) \cup (v_{r+1},N-1]$. The map $f = (f_0, f_1 , \ldots , f_{N-1})^T$ from these variables to $\mathcal{P}$ is the following: 
\begin{equation}
f_n = \begin{cases}
x[n]  \for n \in [0, u_r) \cup (v_{r+1},N-1] \\
\sum_{m=0}^{n-u_{r}} g_{r}[m]h_{r}[n-u_{r}-m] \\
\quad \quad  \for n \in [u_{r}, v_{r}] \\
\sum_{m=0}^{n-u_{r+1}} g_{r+1}[m]h_{r+1}[n-u_{r+1}-m] \\
\quad \quad  \for n \in [u_{r+1}, v_{r+1}].
\end{cases}
\end{equation}

Observe that, for $n=u_{r+1}=v_r$, $f_n$ has two definitions. The variables can admit only those values for which the two definitions have the same value. In the following, we show that there is a one-to-one correspondence between the set of admissible values of the variables and a subset of  $\mathbb{R}^{2N-2}$.   

Each $x[n]$, for $n \in [0, u_r) \cup (v_{r+1},N-1]$, can be chosen from $\subset \mathbb{R}^{2}$. The set of $\{\g_{r+1}, \h_{r+1}\}$ is a subset of $\mathbb{R}^{2 (v_{r+1} - u_{r+1} + 1)}$, which can be seen as follows: $g_{r+1}[l_{g,r+1}-1] = 1$ is fixed (see Lemma \ref{stftproplem}), there are $v_{r+1} - u_{r+1} + 1$ other terms and each can be chosen from $\subseteq \mathbb{R}^2$. 

For each choice of $\{\g_{r+1}, \h_{r+1}\}$, consider the set of $\{\g_{r}, \h_{r}\}$ excluding the terms $h_r[0]$ and $h_r[l_{hr}-1]$: $g_{r}[l_{gr}-1] = 1$ is fixed, there are $v_{r} - u_{r}-1$ other terms and each can be chosen from $\subseteq \mathbb{R}^2$. Hence, this set is a subset of $\mathbb{R}^{2 (v_{r} - u_{r} - 1)}$.

Since the short-time sections $r$ and $r+1$ overlap in the index $v_r$, $\g_r \star \h_r$ and $\g_{r+1} \star \h_{r+1}$  must be consistent in this index, i.e., $\{\g_r, \h_r\}$ must satisfy:
\begin{equation}
\frac{1}{w[0]} h_r[l_{hr}-1] = \frac{1}{w[W-1]}{g_{r+1}[0]}h_{r+1}[0] \label{constraint1}.
\end{equation}
Due to Lemma \ref{stftproplem}, $ \g_r \star \tilde{\h}_r$ and $ \g_{r+1} \star \tilde{\h}_{r+1}$ must also be consistent in this index up to a phase, i.e., $\{\g_r, \h_r\}$ must also satisfy:
\begin{equation}
h_r[0] \equiv \frac{w[0]}{w[W-1]}{g_{r+1}[0]h_{r+1}^\star[l_{h,r+1}-1]} \label{constraint2}.
\end{equation}
Observe that $\equiv$ is used in (\ref{constraint2}), due to the fact that the equality is only up to a phase. However, $h_r[0]$ is real and positive (see Lemma \ref{stftproplem}), due to which (\ref{constraint2}) fixes $h_r[0]$.

Consequently, the set of admissible values of the variables, excluding $h_r[0]$ and $h_r[l_{hr}-1]$, is a subset of $\mathbb{R}^{2N-2}$, as $2( N - v_{r+1} + u_r - 1 + v_{r+1} - u_{r+1} + 1 +  v_{r} - u_{r} - 1) = 2N-2$.  For each point in this set, $h_r[0]$ and $h_r[l_{hr}-1]$ are uniquely determined. It is straightforward to check that the map $f$ from this set to $\mathcal{P}$ is continuously differentiable. Consequently, $\mathcal{Q}_{c}^{rl_rl_{r+1}}$ is the image of a continuously differentiable map from $\mathbb{R}^{2N-2}$ to $\mathcal{P}$.  



}

${(ii) ~ L < W-1:}$




{\edit Consider the setup for which $2L \geq W$. The set of $\{\g_{r+1}, \h_{r+1}\}$, as earlier, is a subset of $\mathbb{R}^{2 (v_{r+1} - u_{r+1} + 1)}$.}

The short-time sections $r$ and $r+1$ overlap in the interval $[u_{r+1}, v_r]$. Let $v_r - u_{r+1} + 1 = T$ (the number of indices in the overlapping interval). Due to $2L \geq W$, we have $T = W-L \leq \lfloor{\frac{W}{2}\rfloor}$. Hence, for each choice of $\{\g_{r+1}, \h_{r+1}\}$, $\{\g_r, \h_r\}$ must satisfy:
\begin{align}
& \sum_{m=0}^{n+u_{r+1}-u_{r}} \frac{1}{w_r[u_{r}+m]} g_r[m]h_r[n+u_{r+1}-u_{r}-m] \label{constraint3} \\
& \hspace{2cm} = \sum_{m=0}^{n} \frac{1}{w_{r+1}[u_{r+1} + m]}{g_{r+1}[m]}h_{r+1}[n-m]  \nonumber
\end{align}
for $0 \leq n \leq T-1$. In addition, $\{\g_r, \h_r\}$ must also satisfy:
\begin{equation}
\frac{1}{w[0]} h_r[0] \equiv \sum_{m=0}^{T-1} \frac{1}{w_{r+1}[u_{r+1} + m]}{g_{r+1}[m]}h_{r+1}[T-1-m]\label{constraint4}
\end{equation}

{\edit If $l_{hr} \geq \lfloor{\frac{W}{2}\rfloor}+1$, then the $T$ bilinear equations (\ref{constraint3}) can be written as $\GG \h_r = \cc$, where $\GG$ has upper triangular structure with unit diagonal entries, due to which $\rank(\GG) = T$. The set of $\g_r$ is a subset of $\mathbb{R}^{2(l_{gr}-1)}$. For each choice of $\g_r$, the terms $\{h_r[l_{hr}-T],\ldots ,h_r[l_{hr}-1]\}$ are fixed by $\GG \h_r = \cc$. The constraint (\ref{constraint4}) fixes the value of $h_r[0]$, as earlier. Each of the remaining $(l_{hr}-1-T)$ terms of $\h_r$ may be chosen from $\subseteq \mathbb{R}^{2}$. 

Hence, the set of admissible values of the variables,  excluding $\{h_r[l_{hr}-T], h_r[l_{hr}-T+1],\ldots ,h_r[l_{hr}-1]\}$ and $h_r[0]$, is a subset of $\mathbb{R}^{2N-2}$, due to the fact that ${2 (N - v_{r+1} + u_r - 1 + v_{r+1} - u_{r+1} + 1 + v_{r} - u_{r} - T )}=2N-2$ (as $l_{gr}+l_{hr}-1 = v_r - u_r + 1$). For each point in this set, $h_r[0]$ and $\{h_r[l_{hr}-T],\ldots ,h_r[l_{hr}-1]\}$ are uniquely determined. The rest of the arguments are identical to those of $L=W-1$.}

If $l_{gr} \geq \lfloor{\frac{W}{2}\rfloor}+1$ instead, then the bilinear equations (\ref{constraint3}) can be equivalently written as $\HH\g_r = \cc$, the same arguments may be applied to draw the same conclusion. For the setup with $2L > W$, the same arguments hold for the short-time sections $r$ and $r+t$, where $t$ is the largest integer such that the short-time sections $r$ and $r+t$ overlap (this ensures $T \leq \lfloor{\frac{W}{2}\rfloor}$). 

   


\section{Proof of Corollary \ref{STFTUS}}

\label{appA2}

We now extend Theorem \ref{STFTUN} to incorporate sparse signals. Let $\mathcal{P}^{S}$ denote the set of all distinct complex signals of length $N$ with a support $S$. Here, $S$ is a binary vector of length $N$, such that $x[n] \neq 0$ whenever $S[n]=1$ and  $x[n] = 0$ whenever $S[n]=0$. Further, $S$ has less than $\min\{L,W-L\}$ consecutive zeros.


Let $\mathcal{P}_c^S \subset \mathcal{P}^S$ denote the set of signals which cannot be uniquely identified from their STFT magnitude if  $\w$ is chosen such that it is non-vanishing and $W \leq \frac{N}{2}$. We show that $\mathcal{P}_c^S$ has measure zero in $\mathcal{P}^S$.

In the proof of Theorem \ref{STFTUN}, in order to show dimension reduction, we used the fact that for sufficient pairs of adjacent short-time sections $r$ and $r+1$, the following holds: 

(i) There is at least one index in the non-overlapping indices $[u_{r}, u_{r+1}-1]$ or $[v_{r}+1, v_{r+1}]$ where the signals $\x_1$ and $\x_2$ have a non-zero value. This ensures that $h_r[0]$ is not constrained by $\{\g_{r+1}, \h_{r+1}\}$ in general. This condition can be ensured by imposing the constraint that the sparse signal cannot have $L$ consecutive zeros.

(ii) There is at least one index in the overlapping indices $[u_{r+1}, v_r]$ where the signals $\x_1$ and $\x_2$ have a non-zero value. This ensures that $h_r[0]$ is constrained by $\{\g_{r+1}, \h_{r+1}\}$ (\ref{constraint2}) for signals which cannot be uniquely identified by their STFT magnitude. This condition can be ensured by imposing the constraint that the sparse signal cannot have $W-L$ consecutive zeros. 

The only difference in the proof is the following: Unlike in the case of non-vanishing signals, there is time-shift ambiguity. Hence, the constraint (\ref{constraint4}) is replaced by:
\begin{equation}
\frac{1}{w[0]} h_r[0] \equiv \sum_{m=0}^{n} \frac{1}{w_{r+1}[u_{r+1} + m]}{g_{r+1}[m]}h_{r+1}[n-m]\label{constraint6}
\end{equation} 
for {\em some} $0 \leq n \leq  T-1$. This fixes the value of $h_r[0]$ to one of at most $T$ values, due to which there is a dimension reduction. 

\section{Proof of Theorem \ref{lhalfthm}}

\label{appC}

We first show the arguments for the case $2W \leq M \leq N$ (short-time autocorrelation known) as the expressions are simple and provide intuition. Then, we show the arguments for the case $4L \leq M < 2W$ (super-resolution).

${(i) ~ 2W \leq M \leq N:}$

The affine constraints in (\ref{STFTPRR}) can be rewritten as (see Section \ref{appD}):
\begin{align}
\nonumber & a_w[ m , r ] = { \sum_{n = 0}^{N-1-m} X[n,n+m] w[ rL - n ]w^\star[rL-(n+m)] }. \nonumber
\end{align}

The proof strategy is as follows: We begin by focusing our attention on short-time section $r=1$. We show that the prior information available, along with the affine autocorrelation measurements corresponding to $r=1$ and the positive semidefinite constraint, will ensure that every feasible matrix of (\ref{STFTPRR}) satisfies $X[n,m]=x_0[n]x_0^\star[m]$ for $0 \leq n , m \leq L$. We then apply this argument incrementally, i.e., we show that the affine measurements corresponding to short-time section $r$, along with the entries of $\X$ uniquely determined  and the positive semidefinite constraint, will ensure that $X[n,m]=x_0[n]x_0^\star[m]$ for $u_r \leq n , m \leq v_r$, where $u_r$ and $v_r$ denote the smallest and largest index where $\w_r$ has a non-zero value respectively. Consequently, the entries along the diagonal and the first $W-L$ off-diagonals of every feasible matrix of (\ref{STFTPRR}) match the entries along the diagonal and the first $W-L$ off-diagonals of the matrix $\x_0\x_0^\star$. Since the entries are sampled from a rank one matrix with non-zero diagonal entries (i.e., $\x_0\x_0^\star$), there is exactly one positive semidefinite completion, which is the rank one completion $\x_0\x_0^\star$ \cite{horn}.

Let $\s_0 = (x_0[0], x_0[1], \ldots, x_0[L])^T$ be a length $L+1$ subsignal of $\x_0$, and $\SSSS$ be the $(L+1) \times (L+1)$ submatrix of $\X$  corresponding to the first $L+1$ rows and columns. We now show that $\SSSS = \s_0 \s_0 ^\star$ is the only feasible matrix under the constraints of (\ref{STFTPRR}).

Since $x_0[n]$ for $0 \leq n \leq \big\lfloor{\frac{L}{2}\big\rfloor}$ is known a priori, we have $S[n,m] = x_0[n]x_0^\star[m]$ for  $0 \leq n,m \leq \big\lfloor{\frac{L}{2}\big\rfloor}$. Let $\mathcal{A}(\SSSS) = \cc$ denote these constraints due to prior information, along with the affine constraints corresponding to $r=1$. In particular, $\mathcal{A}(\SSSS) = \cc$ denotes the following set of constraints:
\begin{align}
& S[n,m] = x_0[n]x_0^\star[m] \quad \textrm{for} \quad 0 \leq n,m \leq \Big\lfloor{\frac{L}{2}\Big\rfloor}, \nonumber \\
& a_w[ m,1 ] = \sum_{n = 0}^{L-m} S[n,n+m]w[L - n ] w^\star[L-(n+m)] \nonumber.
\end{align}

For each feasible matrix $\SSSS$, these set of measurements fix (i) the $\big\lfloor{\frac{L}{2}+1\big\rfloor \times \big\lfloor\frac{L}{2}+1\big\rfloor}$ submatrix, corresponding to the first  $\big\lfloor{\frac{L}{2}+1\big\rfloor}$ rows and columns, of $\SSSS$ (ii) the appropriately weighted sum along the diagonal and each off-diagonal of $\SSSS$ ($2L \leq W$ is implicitly used here). 

\begin{lem}
If $\SSSS_0=\s_0\s_0^\star$ satisfies $\mathcal{A}(\SSSS) = \cc$, then it is the only positive semidefinite matrix which satisfies $\mathcal{A}(\SSSS) = \cc$.
\label{dualcertlem}
\end{lem}
\begin{proof}
Let $T$ be the set of Hermitian matrices of the form 
$$T= \{ \SSSS = \s_0\vv^\star + \vv \s_0^\star : \vv \in \mathbb{C}^n\}$$
and $T^\perp$ be its orthogonal complement.  The set $T$ may be interpreted as the tangent space at $\s_0\s_0^\star$ to the manifold of Hermitian matrices of rank one. Influenced by \cite{candespl}, we use $\SSSS_{T}$ and $\SSSS_{T^\perp}$ to denote the projection of a matrix $\SSSS$ onto the subspaces $T$ and $T^\perp$ respectively.

Standard duality arguments in semidefinite programming show that the following are sufficient conditions for $\SSSS_0=\s_0\s_0^\star$ to be the unique optimizer of (\ref{STFTPRR}):

\begin{enumerate}[(i)]

	\item {\em Condition 1}: $\SSSS \in T \quad \textrm{and} \quad \mathcal{A}(\SSSS) = 0\Rightarrow \SSSS = 0$.  
	
	\item {\em Condition 2}: There exists a {\em dual certificate} $\D$ in the range space of $\mathcal{A}^\star$ obeying:
	\begin{itemize}
	
		\item $\D\s_0 = 0$
		\item $\rank(\D) = L$
		\item $\D \succcurlyeq 0$.
		
	\end{itemize}
	
\end{enumerate}

The proof of this result is based on KKT conditions, and can be found in any standard reference on semidefinite programming (for example, see \cite{vandenberghe}). 


We first show that {\em Condition 1} is satisfied. The set of constraints in $\mathcal{A}(\SSSS) = 0$ due to prior information fix the entries of the first $\big\lfloor{\frac{L}{2}+1\big\rfloor}$ rows and columns of $\SSSS$ to $0$. Since $\SSSS = \s_0\vv^\star + \vv\s_0^\star$ for some $\vv = (v[0],v[1], \ldots, v[L])^T$ (due to $\SSSS \in T$), we infer that $v[n] = icx_0[n]$ for $0 \leq n \leq \big\lfloor{\frac{L}{2}\big\rfloor}$, for some real constant $c$. Indeed, the equations of the form $s_0[n]v^\star[n] + v[n]s_0^\star[n] = 0$ imply $v[n] = ic_n x_0[n]$, for some real constant $c_n$. The equations $s_0[n]v^\star[m] + v[n]s_0^\star[m] = 0$ imply $c_n = c_m$.

The set of constraints in $\mathcal{A}(\SSSS) = 0$ due to the measurements corresponding to $r=1$, along with $v[n] = icx_0[n]$ for $0 \leq n \leq \big\lfloor{\frac{L}{2}\big\rfloor}$, imply $v[n] = icx_0[n]$ for  $\big\lfloor{\frac{L}{2}+ 1\big\rfloor} \leq n \leq L$. Hence, for $\SSSS \in T$, $\mathcal{A}(\SSSS) = 0$ implies $\vv = ic\s_0$, which in turn implies $\SSSS = -ic\s_0 \s_0^\star + ic\s_0 \s_0^\star  =0 $.	

We next establish {\em Condition 2}. For simplicity of notation, we consider the case where $w[n]=1$ for $0 \leq n \leq W-1$. For a general non-vanishing $\w$, the same arguments hold (the Toeplitz matrix considered is appropriately redefined with weights).

The range space of $\mathcal{A}^\star$ is the set of all $L+1 \times L+1$ matrices which are a sum of the following two matrices: The first matrix can have any value in the $\big\lfloor{\frac{L}{2}+1\big\rfloor \times \big\lfloor\frac{L}{2}+1\big\rfloor}$ submatrix corresponding to the first  $\big\lfloor{\frac{L}{2}+1\big\rfloor}$ rows and columns,  and has a value zero outside this submatrix (dual of the set of constraints due to prior information). The second matrix has a Toeplitz structure (dual of the measurements corresponding to $r=1$). 

Suppose $\s_1$ is the vector containing the first $\big\lfloor{\frac{L}{2}+1\big\rfloor}$ entries of $\s_0$ and $\s_2$ is the vector containing the remaining entries of $\s_0$. Here, $\s_1$ corresponds to the locations where we have knowledge of the entries and $\s_2$ corresponds to the locations where the entries are not determined. Let $\LL$ be a lower triangular $\big\lceil{\frac{L}{2}\big\rceil} \times \big\lfloor{\frac{L}{2}+1\big\rfloor}$ Toeplitz matrix satisfying $\LL \s_1 + \s_2 = 0$. Such an $\LL$ always exists if $s_1[0]$ is non-zero and the length of $\s_1$ is greater than or equal to the length of $\s_2$. Let $\Lambda$ be any $\big\lfloor{\frac{L}{2}+1\big\rfloor} \times \big\lfloor{\frac{L}{2}+1\big\rfloor}$  positive semidefinite matrix with rank $ \big\lfloor{\frac{L}{2}\big\rfloor}$ satisfying $\Lambda \s_1 = 0$. Again, such a $\Lambda$ always exists (any positive semidefinite matrix with eigenvectors perpendicular to $\s_1$). Consider the following dual certificate:
\begin{equation}
\D = \begin{bmatrix}
\LL^\star\LL + \Lambda & &\LL^\star\\ 
 &  &  \\
\LL &  & \mathbf{I}_{\big\lceil{\frac{L}{2}\big\rceil}}  \label{dualcertificate}
\end{bmatrix}.
\end{equation}

Clearly, $\D$ is in the range space of $\mathcal{A}^\star$. Also, $\D \s_0 = 0$ by construction. From the Schur complement, it is straightforward to see that $rank(\D) = L$ and $\D \succcurlyeq 0$. 
\end{proof}

We have shown that $\SSSS_0 = \s_0\s_0^\star$ is the only positive semidefinite matrix which satisfies the prior information and the measurements corresponding to $r=1$. Redefine $\s_0$ and $\SSSS$ such that $\s_0 = (x_0[0], x_0[1], \ldots, x_0[2L])^T$ is the $2L+1$ length subsignal of $\x$ and $\SSSS$ is the $(2L+1) \times (2L+1)$ submatrix of $\X$ corresponding to the first $2L+1$ rows and columns. 

We already have $S[n,m] = x_0[n]x_0^\star[m]$ for  $0 \leq n,m \leq L$ from above. Let $\mathcal{A}(\SSSS) = \cc$ denote these constraints, along with the affine constraints corresponding to $r=2$. Due to $2L \leq W$, Lemma \ref{dualcertlem} proves that $\SSSS_0 = \s_0\s_0^\star$ is the only psd matrix which satisfies the prior information and the measurements corresponding to $r=1,2$. Applying this argument incrementally, the entries along the diagonal and the first $W-L$ off-diagonals of every feasible matrix of (\ref{STFTPRR}) match the entries along the diagonal and the first $W-L$ off-diagonals of the matrix $\x_0\x_0^\star$.

{\bf Sparse signals}: The arguments can be seamlessly extended to incorporate sparse signals. 

(i) The fact that there exists a unique positive semidefinite completion once the diagonal and the first $W-L$ off-diagonal entries are sampled from $\x_0\x_0^\star$ holds when $\x_0$ has less than $W-L$ consecutive zeros. 

(ii) Note that the length of $\s_2$ is at most $L$, as it corresponds to the locations in the window where the entries are not determined. Since we know $x_0[n]$ for $i_0 \leq n < i_0 + L$ a priori, where $i_0$ is the smallest index such that $x_0[i_0] \neq 0$, the length of $\s_1$ is $W-L$. Redefine $\s_1$ so that it  corresponds to the locations in the window where the entries are determined, starting from the smallest index which has a non-zero value in order to ensure $s_1[0]\neq 0$. If $\x_0$ has at most $W-2L$ consecutive zeros, then the length of $\s_1$ is at least $(W-L) - (W-2L) = L$. Hence, a lower triangular Toeplitz matrix $\LL$, satisfying $\LL \s_1 + \s_2 = 0$, always exists.

${(ii) ~ 4L \leq M < 2W}$:




The range space of the dual certificate is the set of all $L+1 \times L+1$ matrices which are a sum of the following two matrices: The first matrix can have any value in the $\big\lfloor{\frac{L}{2}+1\big\rfloor \times \big\lfloor\frac{L}{2}+1\big\rfloor}$ submatrix corresponding to the first  $\big\lfloor{\frac{L}{2}+1\big\rfloor}$ rows and columns,  and has a value zero outside this submatrix (dual of the set of constraints due to prior information). The second matrix has the form $\sum_{m=1}^{M} \alpha_m \W_r^\star \f_m\f_m^\star \W_r$, where $\alpha_m$ is real-valued for each $m$ (dual of the measurements corresponding to $r=1$). 

Let ${\bl}= (l[0], l[1], \ldots, l[N-1] )^T$ be a vector that satisfies:
\begin{enumerate}[(i)]
\item $l[0] = 1$, $l[n] = l[N-n] = 0  \for  1 \leq n \leq \left\lceil \frac{L}{2} \right\rceil-1$
\item $\sum_{n=0}^{m} x_0[n]l[m-n] = \sum_{n=0}^{m} x_0^\star[n]l[N - m + n] = 0$ \for $\big\lfloor\frac{L}{2}+1\big\rfloor \leq m \leq L$
\item $\f_m^\star  {\bl} = 0  \for  M + 1 \leq m \leq N$.
\end{enumerate}

These constraints together can be written as $\mathbf{A}{\bl} = \bbb$. When $M \geq 4\lceil \frac{L}{2}\rceil$, the matrix $\mathbf{A}$ is square or wide, and almost always (pseudo) invertible. This can be seen as follows: the determinant of $\mathbf{A}$ is a polynomial function of the entries of $\x_0$, due to which it is either always zero or almost surely non-zero. By substituting $x_0[0]=1$ and $x_0[n]=0$ for $n\neq 0$, it is straightforward to check that the determinant is non-zero. Hence, such an ${\bl}$ almost always exists. 

If the last row in (\ref{dualcertificate}) is chosen as $(l[L], l[L-1], \ldots , l[0])$, then we have: (i) The lower right block is an identity matrix. (ii) $\LL\s_1 + \s_2 = 0$ is satisfied. (iii) Since $\bbb$ is a real vector, ${\bl}$ satisfies $l[n] = l^\star[N-n]$. Therefore,  ${\bl}$ is in the range space of $\sum_{m=1}^{M} \alpha_m \f_m$ where $\alpha_m$ is real-valued, due to which the resulting second matrix is in the range space of $\sum_{m=1}^{M} \alpha_m \W_r^\star \f_m\f_m^\star \W_r$. 

Therefore, $\D$ satisfies all the requirements. The arguments are applied incrementally as earlier, with $M \geq 4L$ for $r > 1$. 




\begin{IEEEbiography}[{\includegraphics[width=1in,height=1.25in]{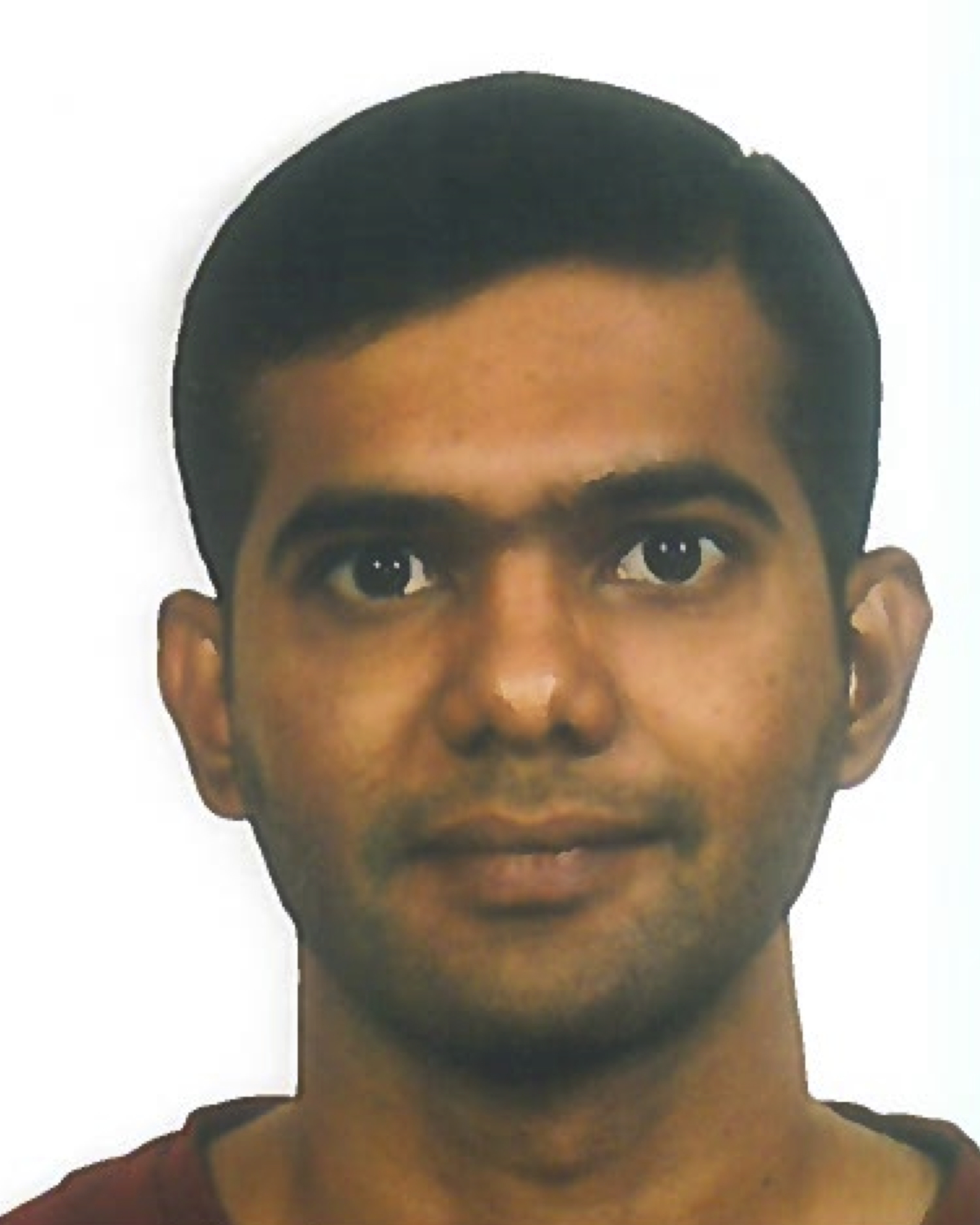}}]{Kishore Jaganathan} was born in Tamil Nadu, India, in 1989. He received the B.Tech degree from Indian Institute of Technology, Madras in 2010, and the M.S. degree from California Institute of Technology, Pasadena in 2011, both in electrical engineering. 

He has been with the California Institute of Technology since 2011, where he is currently pursuing his Ph.D. degree in electrical engineering, under the supervision of Prof. Babak Hassibi. At Caltech, he was awarded the Atwood Fellowship in 2010. He is also the recipient of the Qualcomm Innovation Fellowship, in 2014. His current research interests include signal processing, convex optimization and statistical machine learning.
\end{IEEEbiography}

\begin{IEEEbiography}[{\includegraphics[width=1in,height=1.25in]{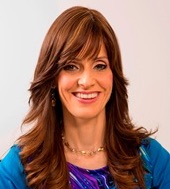}}]{Yonina C. Eldar} (S'98---M'02---SM'07---F'12) received the B.Sc. degree in Physics in 1995 and the B.Sc. degree in Electrical Engineering in 1996 both from Tel-Aviv University (TAU), Tel-Aviv, Israel, and the Ph.D. degree in Electrical Engineering and Computer Science in 2002 from the Massachusetts Institute of Technology (MIT), Cambridge. 

From January 2002 to July 2002 she was a Postdoctoral Fellow at the Digital Signal Processing Group at MIT. She is currently a Professor in the Department of Electrical Engineering at the Technion - Israel Institute of Technology, Haifa, Israel, where she holds the Edwards Chair in Engineering. She is also a Research Affiliate with the Research Laboratory of Electronics at MIT and was a Visiting Professor at Stanford University, Stanford, CA. Her research interests are in the broad areas of statistical signal processing, sampling theory and compressed sensing, optimization methods, and their applications to biology and optics.

Dr. Eldar has received numerous awards for excellence in research and teaching, including the  IEEE Signal Processing Society Technical Achievement Award (2013), the IEEE/AESS Fred Nathanson Memorial Radar Award (2014), and the IEEE Kiyo Tomiyasu Award (2016). She was a Horev Fellow of the Leaders in Science and Technology program at the Technion and an Alon Fellow. She received the Michael Bruno Memorial Award from the Rothschild Foundation, the Weizmann Prize for Exact Sciences, the Wolf Foundation Krill Prize for Excellence in Scientific Research, the Henry Taub Prize for Excellence in Research (twice), the Hershel Rich Innovation Award (three times), the Award for Women with Distinguished Contributions, the Andre and Bella Meyer Lectureship, the Career Development Chair at the Technion, the Muriel \& David Jacknow Award for Excellence in Teaching, and the Technion's Award for Excellence in Teaching (two times).  She received several best paper awards and best demo awards together with her research students and colleagues including the SIAM outstanding Paper Prize, the UFFC Outstanding Paper Award, the Signal Processing Society Best Paper Award and the IET Circuits, Devices and Systems Premium Award, and was selected as one of the 50 most influential women in Israel.

She is a member of the Young Israel Academy of Science and Humanities and the Israel Committee for Higher Education, and an IEEE Fellow. She is the Editor in Chief of Foundations and Trends in Signal Processing, a member of the IEEE Sensor Array and Multichannel Technical Committee and serves on several other IEEE committees. In the past, she was a Signal Processing Society Distinguished Lecturer, member of the IEEE Signal Processing Theory and Methods and Bio Imaging Signal Processing technical committees, and served as an associate editor for the IEEE Transactions On Signal Processing, the EURASIP Journal of Signal Processing, the SIAM Journal on Matrix Analysis and Applications, and the SIAM Journal on Imaging Sciences.  She was Co-Chair and Technical Co-Chair of several international conferences and workshops. She is author of the book {\em Sampling Theory: Beyond Bandlimited Systems} and co-author of the books {\em Compressed Sensing and Convex Optimization Methods in Signal Processing and Communications}, all published by Cambridge University Press.

\end{IEEEbiography}

\begin{IEEEbiography}[{\includegraphics[width=1in,height=1.25in]{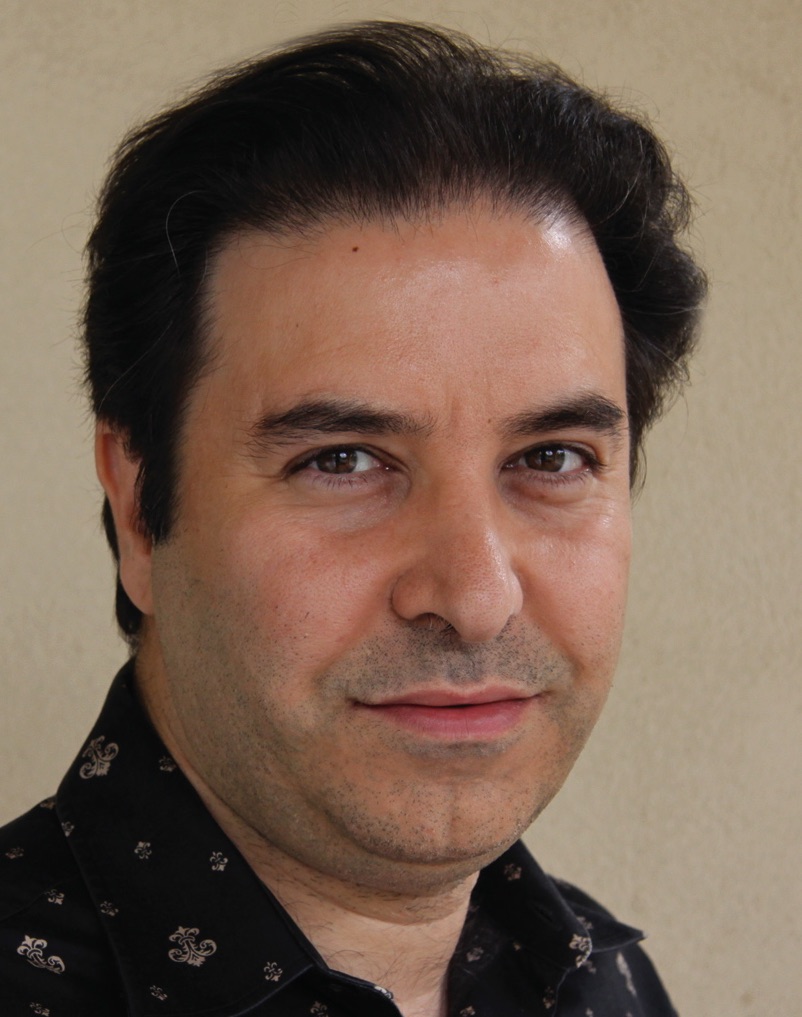}}]{Babak Hassibi} (M'08) was born in Tehran, Iran, in 1967. He received the B.S. degree from the University of Tehran in 1989, and the M.S. and Ph.D. degrees from Stanford University in 1993 and 1996, respectively, all in electrical engineering. 

He has been with the California Institute of Technology since January 2001, where he is currently the Gordon M Binder/ Amgen Professor Of Electrical Engineering. From 2008-2015 he was Executive Officer of Electrical Engineering, as well as Associate Director of Information Science and Technology. From October 1996 to October 1998 he was a research associate at the Information Systems Laboratory, Stanford University, and from November 1998 to December 2000 he was a Member of the Technical Staff in the Mathematical Sciences Research Center at Bell Laboratories, Murray Hill, NJ. His research interests include wireless communications and networks, robust estimation and control, adaptive signal processing and linear algebra. He is the coauthor of the books (both with A.H.~Sayed and T.~Kailath) {\em Indefinite Quadratic Estimation and Control: A Unified Approach to H$^2$ and H$^{\infty}$ Theories} (New York: SIAM, 1999) and {\em Linear Estimation} (Englewood Cliffs, NJ: Prentice Hall, 2000). He is a recipient of an Alborz Foundation Fellowship, the 1999 O. Hugo Schuck best paper award of the American Automatic Control Council (with H.~Hindi and S.P.~Boyd), the 2002 National Science Foundation Career Award, the 2002 Okawa Foundation Research Grant for Information and Telecommunications, the 2003 David and Lucille Packard Fellowship for Science and Engineering,  the 2003 Presidential Early Career Award for Scientists and Engineers (PECASE), and the 2009 Al-Marai Award for Innovative Research in Communications, and was a participant in the 2004 National Academy of Engineering ``Frontiers in Engineering'' program. 

He has been a Guest Editor for the IEEE Transactions on Information Theory special issue on ``space-time transmission, reception, coding and signal processing'' was an Associate Editor for Communications of the IEEE Transactions on Information Theory during 2004-2006, and is currently an Editor for the Journal ``Foundations and Trends in Information and Communication'' and for the IEEE Transactions on Network Science and Engineering. He is an IEEE Information Theory Society Distinguished Lecturer for 2016-2017.
\end{IEEEbiography}

\end{document}